\documentclass{article}

\usepackage{amssymb,amsthm,amsmath,amsfonts}
\newtheorem{theorem}{Theorem}
\newtheorem{remark}[theorem]{Remark}
\newtheorem{example}[theorem]{Example}
\newtheorem{definition}[theorem]{Definition}

\usepackage{colortbl}
\usepackage[table]{xcolor}%

\usepackage{microtype}%

\usepackage{pifont}%
\newcommand{\cmark}{{\color{green}\ding{52}}}%
\newcommand{\xmark}{{\color{red}\ding{56}}}%

\newcommand{\Alice}{Jannik\xspace}
\newcommand{\A}{J\xspace}
\newcommand{\Bob}{Jean-Guillaume\xspace}
\newcommand{\B}{J-G\xspace}
\newcommand{\Carole}{Pascal\xspace}
\newcommand{\C}{P\xspace}
\newcommand{\Dan}{Xavier\xspace}
\newcommand{\X}{X\xspace}

\usepackage{array}
\usepackage{booktabs}
\usepackage{multirow}
\usepackage{diagbox}
\usepackage{pict2e} %
\usepackage{makecell}
\usepackage{xspace}
\usepackage{paralist}
\usepackage{tikz}

\usepackage{algorithm}
\usepackage{algpseudocode}

\algblock{ParFor}{EndParFor}
\algnewcommand\algorithmicparfor{\textbf{parfor}}
\algnewcommand\algorithmicpardo{\textbf{do}}
\newcommand\To{\textbf{to}\ }
\algnewcommand\algorithmicendparfor{\textbf{end\ parfor}}
\algrenewtext{ParFor}[1]{\algorithmicparfor\ #1\ \algorithmicpardo}
\algrenewtext{EndParFor}{\algorithmicendparfor}
\makeatletter
\renewcommand{\ALG@name}{Protocol}
\makeatother

\algnewcommand{\IIf}[2]{\State{#1}\algorithmicif\ #2\ \algorithmicthen}
\algnewcommand{\EndIIf}{\unskip\ \algorithmicend\ \algorithmicif}

\usepackage{amsmath}
\usepackage{eurosym,wasysym}

\DeclareMathOperator{\meq}{\ensuremath{\mathrel{{-}{=}}}}

\DeclareMathOperator*{\argmin}{arg\,min}
\DeclareMathOperator*{\argmax}{arg\,max}

\newcommand{\sharingExpenses}{Shared Expenses}
\newcommand{\sharingExpensesProblem}{\sharingExpenses{} Problem}
\newcommand{\SEP}{SEP}

\newcommand{\NP}{\ensuremath{\mathcal{NP}}}

\usepackage[colorinlistoftodos]{todonotes}

\usepackage{color,svgcolor}
\newcommand{\thetitle}{A Faster Cryptographer's Conspiracy Santa}
\makeatletter

\usepackage[plainpages=true]{hyperref}
\hypersetup{
pdftitle={\thetitle},
pdfauthor={X. Bultel, J. Dreier, J.-G. Dumas, and P. Lafourcade},
breaklinks=true, %
colorlinks=true,
 linkcolor=darkred,
 citecolor=blue,
 urlcolor=darkgreen,
}
\makeatother

\newcommand{\OpenDreamKit}{the \href{http://opendreamkit.org}{OpenDreamKit} \href{https://ec.europa.eu/programmes/horizon2020/}{Horizon 2020} \href{https://ec.europa.eu/programmes/horizon2020/en/h2020-section/european-research-infrastructures-including-e-infrastructures}{European Research Infrastructures} project (\#\href{http://cordis.europa.eu/project/rcn/198334_en.html}{676541})}

\newcommand{\funding}{This research was conducted with the support of the FEDER program of
2014-2020, the region council of Auvergne-Rh\^one-Alpes, the support
of the ``Digital Trust'' Chair from the University of Auvergne
Foundation, the Indo-French Centre for the Promotion of Advanced
Research (IFCPAR), the Center Franco-Indien Pour La Promotion De La
Recherche Avanc\'ee (CEFIPRA) through the project DST/CNRS 2015-03
under DST-INRIA-CNRS Targeted Programme, 
and \OpenDreamKit.}

\newcommand{\acknowledgements}{Many thanks to Marie-B\'eatrice,
  Anne-Catherine, Marc, Jacques and Luc for such great conspiracy
  Santas! A big thanks also to the Gilbert family for having big
  instances of the \sharingExpensesProblem{} problem regularly, and to
  Cyprien for asking the question of its complexity. More thanks go to
  Mathilde and Gw\'ena\"el for the discussions on efficient algorithms.}

\begin{document}

\date{}

\title{\thetitle\footnote{\funding}}

\author{Xavier Bultel\footnotemark[2]
\and Jannik Dreier \footnotemark[3]
\and Jean-Guillaume Dumas \footnotemark[4]
\and Pascal Lafourcade\footnotemark[5]}
\maketitle

\footnotetext[2]{LIFO, INSA Centre Val de Loire, Universit\'e d'Orl\'eans,
88 boulevard Lahitolle
CS 60013 - 18022 Bourges, cedex, France}
\footnotetext[3]{Université de Lorraine, CNRS, Inria, LORIA, F-54000
  Nancy, France}
\footnotetext[4]{Universit\'e Grenoble Alpes, Laboratoire Jean Kuntzmann,
  UMR CNRS 5224, 700~avenue centrale, IMAG - CS 40700, 38058 Grenoble,
  cedex 9, France}
\footnotetext[5]{University Clermont
  Auvergne, LIMOS, CNRS UMR 6158, Campus Universitaire des Cézeaux, 1
  rue de la Chebarde, 63170 Aubi\`ere, France}

\begin{abstract}
In Conspiracy Santa, a variant of Secret Santa, a group of people
offer each other Christmas gifts, where each member of the group
receives a gift from the other members of the group. 
To that end, the members of the group form conspiracies, to decide on
appropriate gifts, and usually divide the cost of each gift among all
participants of that conspiracy.
This requires to settle the shared expenses per conspiracy, 
so Conspiracy Santa can actually be seen as an aggregation of 
several shared expenses problems.

First, we show that the problem of finding a minimal number of
transaction when settling shared expenses is NP-complete. Still, there
exist good greedy approximations.  Second, we present a greedy
distributed secure solution to Conspiracy Santa. This solution allows
a group of $n$ people to share the expenses for the gifts in such a
way that no participant learns the price of his gift, but at the same
time notably reduces the number of transactions to $2 \cdot n+1$ with respect
to a na\"ive aggregation of $n \cdot (n-2)$.  Furthermore, our solution does
not require a trusted third party, and can either be implemented
physically (the participants are in the same room and exchange money
using envelopes) or, over Internet, using a cryptocurrency.

\end{abstract}



\section{Introduction}

\emph{Secret Santa} is a Christmas tradition, where members of a group
are randomly assigned to another person, to whom they have to offer a
gift. The identity of the person offering the present is usually
secret~\cite{Liberti:2008:secretsanta,Mauw:2014:ssanta,Ryan2016}, as well as the price of the present. Moreover, the participants
often determine a common bound for the gift's prices.

In \emph{Conspiracy Santa}, a variant of Secret Santa, for each
participant, the other members of the group collude and jointly decide
on an appropriate gift.  The gift is then usually bought by one of the
colluding participants, and the expenses are shared among the
colluding participants.

In this setting, the price of the gift must remain secret and,
potentially, also who bought the present.  At the same time, sharing
the expenses usually results in numerous transactions. Existing
results in the literature
(e.g.,~\cite{settleup,splitwise,tricount,Thesis12}) aim at minimizing
the number of transactions, but they assume that all expenses are
public, that all participants are honest, and that communications are
safe. Our goal is instead to propose a secure Conspiracy Santa algorithm for
paranoid cryptographers that do not want to disclose the prices.
Further, they want to keep their privacy at all cost, so they might rely on
external third parties but only in case they do not need to trust them.

\subsection{Contributions}

Our results can be split into the following $3$  contributions:
\begin{itemize}
\item We show that the general problem of finding a solution with a
  minimal number of transactions when sharing expenses
  (\sharingExpensesProblem{}, or \SEP{}) is NP-complete.
\item We provide secure protocols for Conspiracy Santa for arbitrary
  $n$ participants. The algorithms ensure that no participant learns
  the price of his gift, nor who bought it.  Moreover, the algorithms
  reduce the number of transactions necessary to $3 \cdot n$ or $2 \cdot n+1$
  (depending on the largest authorized amount for a given transaction)
  compared to a na\"ive solution of $n \cdot (n-2)$.
\item Our secure algorithms are entirely distributed and do not
  require any trusted third party. To also realize the payments in a
  distributed fashion, a secure peer-to-peer cryptocurrency can be
  used.  Additionally, we present a physical payment solution where all
  participants need to be in the same place, using envelopes and bank
  notes.
\end{itemize}
Our algorithms can also be used in the case where expenses are shared
within multiple groups. There, some people belong to several of these
groups and the goal is to reduce the number of transactions while still
ensuring privacy: all participants only learn about the expenses
of their groups, not the other groups. 
One can also see this problem as a variant of the dining
cryptographers~\cite{Chaum1988}. However, instead of respecting the
cryptographers' right to anonymously invite everybody, we here want to
respect the cryptographers' right to privately share expenses of multiple
diners with different groups.

Table~\ref{tab:algos} summarizes the number of transactions, as well
as the order of magnitude of the largest amount per transaction, required
for the different algorithms considered in order to realize a
conspiracy Santa.

\begin{table}[htbp]\centering
\begin{tabular}{|l|c|r|r|}
\hline
\multirow{2}{*}{Algorithm} & \multirow{2}{*}{Peer-to-peer} &  \multirow{2}{*}{Transactions} & Largest amount \\
 &  &   & per transaction \\
\hline
$n$ instances of \SEP{} & \cmark & $n\cdot(n-2)$ & constant\\
With a trusted third-party & \xmark & $n$ & constant\\
Here (Protocols~\ref{alg:setup}, \ref{alg:firstround},
  \ref{alg:secondround}) & \cmark & $3 \cdot n$ & constant\\
Here (Protocols~\ref{alg:setup}, \ref{alg:merged},
  \ref{alg:onlythird}) & \cmark & $2 \cdot n+1$ & linear\\
\hline
\end{tabular}
\caption{Number of transactions for Conspiracy Santa}\label{tab:algos}
\end{table}

With respect to the conference version of this
paper~\cite{fun:2018:conspiracy}, we here provide the faster algorithm
requiring only $2\cdot n+1$ transactions, instead of $3\cdot n$. 
With respect to an upper bound $B$ on the price of any gift, in this faster
protocol, the amount of each transaction can then grow up to $n\cdot B$,
where it stayed below $B$ in our other protocol.
We also provide a physical variant as well as complexity and security
proofs for this novel protocol.

\subsection{Outline}

The remainder of the paper is structured as follows: in
Section~\ref{sec:complexity}, we analyze the complexity of the general
problem of sharing expenses.  In Section~\ref{sec:conspiracy}, we
present our protocol to solve the problem of privately sharing
expenses in Conspiracy Santa, in a peer-to-peer setting.  We also
discuss further applications of our solution, and how to realize the
anonymous payments required by the algorithm, either physically or
online.  In Section~\ref{sec:faster}, we then present our second
algorithm, with less transactions.
We finally conclude in Section~\ref{sec:conclusion}.%

\section{The \sharingExpensesProblem{} and its Complexity}\label{sec:complexity}

Before analyzing the Conspiracy Santa problem in more detail, we first
discuss the problem of settling shared expenses with a
minimal number of transactions.  This problem frequently arises, for
example when a group of security researchers attends a FUN conference
and wants to share common expenses such as taxis, restaurants etc.
Reducing the overall number of transactions might then reduce the
overall currency exchange fees paid by the researchers.

In such a case, each participant covers some of the common expenses, and in the end of the conference, some transactions are necessary to ensure that all participants payed the same amount.
Note for this first example, there are no privacy constraints, as all amounts are public.
\begin{example}\label{ex:fun16}
\Alice, \Bob, and \Carole attended FUN'16. The first night, \Alice payed
the restaurant for $155$ \euro, and \Bob the drinks at the bar for $52$
\euro. The second day \Carole payed the restaurant and drinks for a
total of $213$ \euro.

The total sum is then $155+52+213=420$ \euro, meaning $140$ \euro{} per
person.  This means that \Alice payed $140-155=-15$ \euro{} too much,
\Bob needs to pay $140-52 = 88$ \euro{} more, and \Carole has to receive
$140-213 = -73$ \euro.  In this case, the optimal solution uses two
transactions: \Bob gives $15$ \euro{} to \Alice, and $73$ \euro{} to
\Carole.
\end{example}
There are numerous applications implementing solutions to this problem (e.g., \cite{settleup,splitwise,tricount}), but it is unclear how they compute the transactions.
Moreover, in these applications all expenses are public, making them unsuitable for Conspiracy Santa.

David V{\'a}vra wrote a master's thesis~\cite{Thesis12} about a
similar smartphone application that allows to settle expenses within a
group.
He discusses a greedy approximation algorithm (see below), and
conjectures that the problem is \NP-complete, but without giving a
formal proof. We start by formally defining the problem.
\begin{definition}{\sharingExpensesProblem{} (\SEP{}).}
  \begin{description}
  \item[Input:] Given a multiset of values $K = \{k_1, \ldots, k_n\}$ such that
$\sum_{i=1}^{n} k_i = 0$, where a positive $k_i$ means that
participant $i$ has to pay money, and a negative $k_i$ means that $i$
has to be reimbursed. 
\item[Question:] Is there a way to do all reimbursements using (strictly) less than $n-1$ transactions?
  \end{description}
\end{definition}
Note that there is always a solution using $n-1$ transactions using a
greedy approach: given the values in $K = \{k_1, \ldots, k_n\}$, let
$i$ be the index of the maximum value of $K$ ($i= \argmax_i(k_i)$) and
let $j$ be the index of the minimum value of $K$ ($j=
\argmin_j(k_j)$), we use one transaction between $i$ and $j$ such that
after the transaction either the participant $i$ or $j$ ends up at
$0$. I.e., if $|k_i|-|k_j|>0$, then the participant $j$ ends up at $0$,
otherwise the participant $i$ ends up at $0$. 
By then recursively applying
the same procedure on the remaining $n-1$ values, we can do all
reimbursements. Overall, this greedy solution uses
$n-1$ transactions in the worst case.

It is easy to see that $\SEP \in \NP$: guess a list of (less than $n-1$) transactions, and verify for each participant that in the end there are no debts or credits left.

We show that \SEP{} is \NP{}-complete, for this we use a reduction
from the \emph{Subset Sum Problem}~\cite{Karp2010} which can be seen as
a special case of the well known knapsack
problem~\cite{Garey:1990:CIG:574848}.
\begin{definition}{Subset Sum Problem (SSP)}
\begin{description}
\item[Input:] Given a multiset of values $K = \{k_1, \ldots, k_n\}$.
\item[Question:] Is there a subset $K' \subseteq K$ such that
  $\sum_{k' \in K'} k' = 0$?
\end{description}
\end{definition}
The Subset Sum Problem is known to be \NP{}-complete (see, e.g., \cite{Cormen}).
\begin{theorem}{}
The \sharingExpensesProblem{} is \NP{}-complete.
\end{theorem}
\begin{proof}
Consider the following reduction algorithm:

Given a Subset Sum Problem (SSP) instance, i.e., a multiset of values
$K = \{k_1, \ldots, k_n\}$, compute $s = \sum_{k \in K} k$.  If $s =
0$ then return yes, otherwise let $K' = K \cup \{-s\}$ and return the
answer of an oracle for the \sharingExpensesProblem{} for $K'$.

It is easy to see that the reduction is polynomial, as computing the
sum is in $\mathcal{O}(n)$. We now need to show that the reduction is
correct. We consider the two following cases:

\begin{itemize}
\item  Suppose the answer to the SSP is yes, then there is a subset $K''
\subseteq K$ such that $\sum_{k \in K''} k = 0$.  If $K'' = K$, then
the check in the reduction is true, and the algorithm returns yes.  If
$K'' \neq K$, then we can balance the expenses in the sets $K''$ and
$K' \setminus K''$ independently using the greedy algorithm explained
above.  This results in $|K''| - 1$ and $|K'| - |K''| -1$ transactions
respectively, for a total of $|K'| - |K''| - 1 + |K''| - 1 = |K'|-2 <
|K'|-1$ transactions.  Thus there is a way to do all reimbursements
using strictly less than $|K'|-1$ transactions, hence the answer will be yes.
\item Suppose the answer to the SSP is no, then there is no subset
  $K'' \subseteq K$ such that $\sum_{k \in K''} k = 0$.  This means
  that there is no subset $K_3 \subseteq K'$ such that the expenses
  within this set can be balanced independently of the other expenses.
  To see this, suppose it were possible to balance the expenses in
  $K_3$ independently, then we must have $\sum_{k \in K_3} k = 0$,
  contradicting the hypothesis that there is no such subset (note that
  w.l.o.g. $K_3 \subseteq K$, if it contains the added value one can
  simply choose $K' \setminus K_3$).

Hence any way of balancing the expenses has to involve all $n$
participants, but building a connected graph with $n$ nodes requires
at least $n-1$ edges.  Thus there cannot be a solution with less than
$n-1$ transactions, and the oracle will answer no. 
\end{itemize}
 \qed\end{proof}

\section{Cryptographer's Conspiracy Santa}\label{sec:conspiracy}

Consider now the problem of organizing Conspiracy Santa, where no participant shall learn the price of his gift.
Obviously we cannot simply apply, e.g., the greedy algorithm explained above on all the expenses, as this would imply that everybody learns all the prices.

More formally, an instance of Conspiracy Santa with $n$ participant 
consists of $n$ shared expenses problem (sub-\SEP), each with $n-1$
participants and with non-empty intersections of the participants. 
In each sub-\SEP, the
$n-1$ participants freely discuss, decide on a gift, its value $v_i$ and who
pays it; then agree that their share for this gift is
$v_i/(n-1)$. Overall the share of each participant $j$ is
$$\frac{\sum_{i=1,i\neq{}j}^n v_i}{n-1}.$$
A participants \emph{balance} $p_j$ is this share minus the values of the gifts he bought.

A simple solution would be to use a trusted third party, but most
cryptographers are paranoid and do not like trusted third parties.
A distributed solution would be to settle the expenses for each gift
within the associated conspiracy group individually, but this then
results in $n$ instances of the problem, with $n-2$ transactions each
(assuming that only one person bought the gift), for a total of
$n\cdot (n-2)$ transactions.

Moreover, the problem becomes more complex if several groups
with non-empty intersections want to minimize transactions all
together while preserving the inter-group privacy. 

\begin{example}\label{ex:group}
  \emph{Example~\ref{ex:fun16} continued.}
  For the same conference, FUN'16, \Alice, \Bob and \Dan
  shared a taxi from the airport and \Bob paid for a total of
  $60$\euro, that is $20$\euro{} per person. 
  There are two possibilities. Either \Alice and \Dan make two new
  transactions to reimburse \Bob. Or, to minimize the overall number of
  transactions, they aggregate both accounts, i.e. those from
  Example~\ref{ex:fun16} with those of the  
  taxi ride. That is 
  $[-15,88,-73,0]+[20,-40,0,20]=[5,48,-73,20]$.
  Overall \Alice thus gives $5$ \euro{} to
  \Carole, \Bob reduces his debt to \Carole to only $48$\euro{} and
  \Dan gives $20$ \euro{} to \Carole.
  The security issue, in this second case, is that maybe \Alice and
  \Bob did not want \Dan to
  know that they were having lunch with \Carole, nor that they had a
  debt of more than $20$ \euro, etc.
\end{example}
In the next part, we present our solution for the generalization of
Conspiracy Santa as the aggregation of several shared expenses
problems with non-empty intersections between the participants. This
solution uses $3 \cdot n$ transactions, preserves privacy, and does not
require a trusted third party.

\subsection{A Distributed Solution using Cryptocurrencies}

\paragraph*{Assumptions}
\begin{enumerate}
\item We suppose that all participants know a {\bf fixed upper bound} $B$ for the
value of any gift.  
\item We suppose that each participant \textbf{balance is an integral
  number}. A simple solution for this is to express everything in cents
  (\cent), and make users agree that shares could be unevenly
  distributed up to a difference of one cent.
\item We consider \textbf{semi-honest} participants in the sense that the participants follow \emph{honestly} the protocol, but they try to exploit all intermediate information that they have received during the protocol to break privacy.
\end{enumerate}

Apart from the setup, the protocol has $3$ rounds,
each one with $n$ transactions, and one initialization phase.

\paragraph*{Initialization Phase}
In the setup phase, the participants learn the price of the gifts in
which they participate and can therefore compute their overall
balance, $p_i$. They also setup several anonymous addresses in a given
public transaction cryptocurrency like Bitcoin~\cite{bitcoin},
ZCash~\cite{zcash} or Monero~\cite{monero}.

Finally the participants create one anonymous address which is used as
a piggy bank. They all have access to the secret key associated to
that piggy bank address. For instance, they can exchange encrypted
emails to share this secret
key. Protocol~\ref{alg:setup} presents the details of this setup
phase.

\begin{algorithm}[htb]
\caption{\SEP{} broadcast setup}\label{alg:setup}
\begin{algorithmic}[1]
  \Require An upper bound $B$ on the value of any gift;
  \Require All expenses.
  \Ensure Each participant learns his balance $p_i$.
  \Ensure Each participant creates $1$ or several anonymous currency
  addresses.
  \Ensure A shared anonymous currency address.
  \State One anonymous currency address is created and the associated
  secret key is shared among all participants.
  \For{each exchange group}
	\For{each payment within the group}
		\State broadcast the amount paid to all members of the group;
	\EndFor
  	\For{each participant in the group}
		\State Sum all the paid amounts of all the participants;
		\State Divide by the number of participants in the group;
		\State This produces the in-group share by participant.
	\EndFor
  \EndFor
  \For{each overall participant}
	\State Add up all in-group shares;
	\State Subtract all own expenses to get $p_i$;
	\If{$p_i<0$}
		\State Create $\lfloor\frac{p_i}{B}\rfloor$ anonymous
                currency addresses.
        \EndIf
  \EndFor
\end{algorithmic}
\end{algorithm}

\paragraph*{First Round}
The idea is that the participants will round their debts or credits so
that the different amounts become indistinguishable.  For this, the
participants perform transactions to adjust their balance to either
$0$, $B$ or a negative multiple of $B$.  The first participant
randomly selects an initial value between $1$ and $B$~\cent{}, and
sends it to the second participant.  This transaction is realized via
any private payment channel between the two participants. It can be a
physical payment, a bank transfer, a cryptocurrency payment, \ldots,
as long as no other participant learns the transferred amount.  Then
the second participant adds his balance to the received amount modulo
$B$, and forwards the money\footnote{up to $B$, or such that its
  credit becomes a multiple of $B$.} to the next participant, and so
on.  The last participant also adds his balance and sends the
resulting amount to the first participant.  In the end, all
participants obtain a balance of a multiple of $B$, and the random
amount chosen by the first participant has hidden the exact amounts.
The details are described in Protocol~\ref{alg:firstround}.

\begin{algorithm}[htb]
\caption{Secure rounding to multiple of the bound}\label{alg:firstround}
\begin{algorithmic}[1]
  \Require An upper bound $B$ on the value of any gift;
  \Require Each one of $n$ participants knows his balance $p_i$;
  \Require $\sum_{i=1}^n p_i=0$.
  \Ensure Each one of $n$ participants has a new balance $p_i$, either
  $0$, $B$ or a negative multiple of $B$;
  \Ensure  $\sum_{i=1}^n p_i=0$;
  \Ensure Each transaction is between $1$ and $B$ \cent;
  \Ensure The protocol is zero-knowledge.
  \State{$P_1$}: $t_1\stackrel{\$}{\longleftarrow}\left[1..B\right]$
  uniformly sampled at random;
  \State{$P_1$}: $p_1 = p_1 - t_1$;
  \State $P_1$ sends $t_1$ \cent{} to $P_2$; \hfill \Comment{Random
    transaction $1..B$ on a secure channel} \label{i:tfromp}
  \State{$P_2$}: $p_2 = p_2 + t_1$;
  \For{$i=2$ \To $n-1$}
  	\State{$P_i$}: $t_i = p_i \mod B$; \label{i:piti}
	\IIf{$P_i$: }{$t_i = 0$}{ $t_i = t_i + B$;}\EndIIf\hfill\Comment{$1 \leq t_i \leq B$} 
	\State{$P_i$}: $p_i = p_i - t_i$;
	\State $P_i$ sends $t_i$ \cent{} to $P_{i+1}$; \hfill
        \Comment{Random transaction $1..B$ on a secure channel}\label{i:tfrompn}
	\State{$P_{i+1}$}: $p_{i+1} = p_{i+1} + t_i$;
  \EndFor
  \State{$P_n$}: $t_n = p_n \mod B $;
  \IIf{$P_n$: }{$t_n = 0$}{ $t_n = t_n + B$;}\EndIIf\hfill\Comment{$1 \leq t_n \leq B$}
  \State{$P_n$}: $p_n = p_n - t_n$;
  \State $P_n$ sends $t_n$ \cent{} to $P_1$; \hfill \Comment{Random transaction $1..B$ on a secure channel} \label{i:ttop}
  \State{$P_1$}: $p_1 = p_1 + t_n$;
\end{algorithmic}
\end{algorithm}

\paragraph*{Second Round}
The second and third rounds of the protocol require anonymous
payments, for which we use anonymous cryptocurrency addresses. These
two rounds are presented in Protocol~\ref{alg:secondround}, where
$parfor$ is a notation for parallel execution of a $for$ command where
the order is not important.
\begin{algorithm}[htb]
\caption{Peer-to-peer secure debt resolution}\label{alg:secondround}
\begin{algorithmic}[1]
  \Require An upper bound $B$ on the value of any gift;
  \Require $n$ participants each with a balance $p_i$, either
  $0$, $B$ or a negative multiple of $B$.
  \Ensure All balances are zero;
  \Ensure The protocol is zero-knowledge.
  \ParFor{$i=1$ \To $n$}\Comment{Everybody sends $B$ to the piggy bank}
	\State{$P_i$}: $p_i \meq B$;
        \State $P_i$ sends $B$ \cent{} to the shared anonymous address; \label{i:put}
        	\hfill \Comment{Public transaction of $B$}
  \EndParFor
  \ParFor{$i=1$ \To $n$}
	\If{$p_i<0$}\Comment{Creditors recover their assets}
  		\ParFor{$j=1$ \To $\frac{-p_i}{B}$}
        		\State $P_i$ makes the shared anonymous address
                        pay $B$\cent{} to one of his own anonymous
                        addresses; 
                        		\hfill \Comment{Public transaction of $B$}\label{i:take}
  		\EndParFor
        	\State{$P_i$}: $p_i = 0$.
        \EndIf
  \EndParFor
\end{algorithmic}
\end{algorithm}
In the second round, every participant makes one public transaction of
$B$ \cent{} to the piggy bank.

\paragraph*{Third Round}

Each creditor recovers their assets via
$\lfloor\frac{p_i}{B}\rfloor$ public transactions of $B$ \cent{} from
the piggy bank. Note that if a participant needs to withdraw more than
$B$ \cent{} he needs to perform several transactions.  To ensure anonymity,
he needs to use a different anonymous address for each transaction.
In the end, the account is empty and the number of transactions
corresponds exactly to the number of initial transactions used to
credit the piggy bank's account.

\begin{theorem}\label{thm:threen} For $n$ participants,
  Protocols~\ref{alg:setup}, \ref{alg:firstround},
  \ref{alg:secondround} are correct and, apart from the setup, require
  $3\cdot n$ transactions. 
\end{theorem}
\begin{proof}
Including the piggy bank, all the transactions are among participants,
therefore the sum of all the debts and credits is invariant and
zero. There remains to prove that in the end of the protocol all the
debts and credits are also zero.  The value of any gift is bounded by
$B$, thus any initial debt for any gift is at most $B/(n-1)$. As
participants participate to at most $n-1$ gifts, the largest debt is
thus lower than $B$ \cent{}.  Then, during the first round, all
participants, except $P_1$, round their credits or debts to multiples
of $B$. But then, by the invariant, after the first round, the debt or
credit of $P_1$ must also be a multiple of $B$.  Furthermore, any
debtor will thus either be at zero after the first round or at a debt
of exactly $B$ \cent{}. After the second round any debtor will then be
either at zero or at a credit of exactly $B$ \cent{}. Thus after the
second round only the piggy bank has a debt. Since the piggy bank
received exactly $n \cdot B$ \cent{}, exactly $n$ transactions of $B$ \cent{}
will make it zero and the invariant ensures that, after the third
round, all the creditors must be zero too. \hfill \qed\end{proof}

\begin{remark}
It is important to use a cryptocurrency such as Bitcoin, Monero or
ZCash in order to hide both the issuer and the receiver of each
transaction in the third round. This ensures that nobody can identify the users.

Note that when using Bitcoin, users can potentially be tracked if the
addresses are used for other transactions.  Using Monero or Zcash can
offer more privacy since the exchanged amount can also be anonymized.
Moreover, to avoid leaking the fact that some persons need to
withdraw $B$\cent{} multiple times, and are thus doing multiple
transaction at the same time, all the withdrawals should be
synchronized. If exact synchronization is difficult to achieve, one
can decide on a common time interval, e.g., an hour, and all the
transactions have to be done at random time points during this
interval, independently, whether they are executed from the same or a
different participant.
\end{remark}

\begin{example}\label{ex:algo}
We now have a look at the algorithm for our example with \Alice, \Bob,
\Carole and \Dan. As in Example~\ref{ex:group}, the initial balance
vector is $[5,48,-73,20]$\footnote{In our running example the values are integral numbers in euros, so for simplicity we continue in euros instead of
cents as in the protocol description.}.
They decide on an upper bound of
$B=50$ \euro{} (note that to provably \emph{ensure} exactly $3\cdot n=12$ transactions
they should take an upper bound larger than any expense, that is
larger than $213$ \euro, but $50$ is sufficient for our example here).
For the first round, \Alice randomly selects $1 \leq{} t_1=12 \leq{} 50$ and makes
a first private transaction of $t_1=12$ \euro{} to \Bob.  \Bob then
makes a private transaction of $t_2=12+48 \mod 50=10$ \euro{} to
\Carole; \Carole makes a private transaction of $t_3=10-73 \mod 50 =
37$ \euro{} to \Dan; who makes a private transaction of $t_4=37+20 \mod
50 = 7$ \euro{} to \Alice. All these transactions are represented in
Figure~\ref{fig:1round}. The balance vector is thus now
$[0,50,-100,50]$, because for instance \Bob had a balance of
$48$ \euro{}, received $12$ \euro{} from \Alice and sends
$10$ \euro{} to \Carole, hence his new balance is $48+12-10=50$~\euro{}.
Everybody sends $50$~\euro{} to the piggy bank address, so that the
balance vector becomes $[-50,0,-150,0]$.  Finally there are four
$50$~\euro{} transactions, one to an address controlled by \Alice and
three to (different) addresses controlled by \Carole. These two last
rounds are illustrated in Figure~\ref{fig-2round}.
Note that we have exactly $n=4$ transactions per round.
\end{example}

\begin{figure}[htb]
\begin{center}
  \begin{tikzpicture}
    \node (A) at (0, 3) {\A: 5};
    \node (B) at (3, 3) {\B: 48};
    \node (C) at (3, 0) {\C: -73};
    \node (D) at (0, 0) {\X: 20};
    \draw[thick,-latex] (A) .. controls (1,4) and (2,4) .. (B);
    \draw[thick,-latex] (C) .. controls (2,-1) and (1,-1) .. (D);
    \draw[thick,-latex] (D) .. controls (-1,1) and (-1,2) .. (A);
    \draw[thick,-latex] (B) .. controls (4,2) and (4,1) .. (C);
    
    \node (R) at (-0.9,3.3) {$t_1 = 12$};
    \node (AB) at (1.6,4) {$12$ \euro{}};
    \node (BC) at (5.5,1.5) {$12+48=10 \mod 50$};
    \node (CD) at (1.5,-1) {$10-73=37 \mod 50$};
    \node (DA) at (-2.5,1.5) {$37+20=7 \mod 50$};
    
  \end{tikzpicture}
\end{center}
	\caption{First round (private transactions) of Example~\ref{ex:algo} (starting at $t_1=12$,
  ending with $5-12+7=0\mod 50$).}\label{fig:1round}
\end{figure}
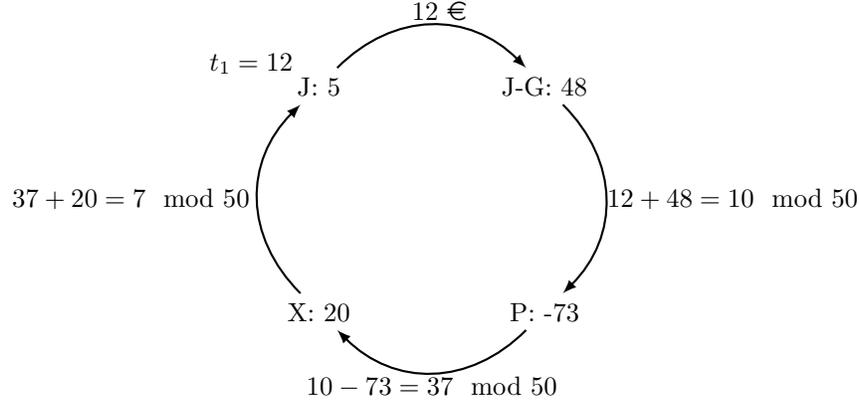

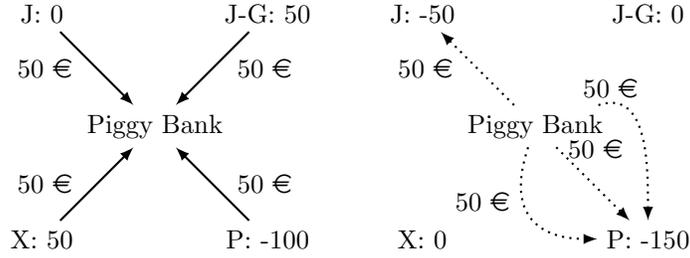
\begin{figure}[htb]
\begin{center}
  \begin{tikzpicture}
    \node (A) at (0, 3) {\A: 0};
    \node (B) at (3, 3) {\B: 50};
    \node (C) at (3, 0) {\C: -100};
    \node (D) at (0, 0) {\X: 50};
    \node (BA) at (1.5,1.5) {Piggy Bank};
    
    \draw[->,thick,-latex] (A) -> (BA) node[midway,left=2mm] {50 \euro};
    \draw[->,thick,-latex] (B) -> (BA) node[midway,right=2mm] {50 \euro};
    \draw[->,thick,-latex] (C) -> (BA) node[midway,right=2mm] {50 \euro};
    \draw[->,thick,-latex] (D) -> (BA) node[midway,left=2mm] {50 \euro};
    
  \end{tikzpicture}
  ~~~~~
  \begin{tikzpicture}
    \node (A) at (0, 3) {\A: -50};
    \node (B) at (3, 3) {\B: 0};
    \node (C) at (3, 0) {\C: -150};
    \node (D) at (0, 0) {\X: 0};
    \node (BA) at (1.5,1.5) {Piggy Bank};
    
    \node (L1) at (2.5,2) {50 \euro};
    \node (L2) at (0.8,0.5) {50 \euro};
    \node (L3) at (2.3,1.2) {50 \euro};
     
    \draw[->,thick,dotted,-latex] (BA) -> (A) node[midway,left=2mm] {50 \euro};
    \draw[->,thick,dotted,-latex] (BA) .. controls (3,2) and (3,1) .. (C);
    \draw[->,thick,dotted,-latex] (BA) to (C);
    \draw[->,thick,dotted,-latex] (BA) .. controls (1,0) and (2,0) .. (C);
    
  \end{tikzpicture}
\end{center}
\caption{On the left: second round of Example~\ref{ex:algo}. On the right: third round of Example~\ref{ex:algo}. Dotted arrows represent anonymous transactions, in particular \Carole uses three different anonymous addresses.}\label{fig-2round}
\end{figure}

\subsection{Security Proof} 
\newcommand{\D}{\mathcal{D}}
\newcommand{\I}{\mathcal{I}}

We now provide a formal security proof for our protocol.   
We use the standard multi-party computations definition of security
against semi-honest adversaries~\cite{Ma2008}. 
As stated above, we consider \emph{semi-honest} adversaries in the sense that the entities run \emph{honestly} the protocols, but they try to exploit all
intermediate information that they have received during the protocol.

We start by formally defining the \emph{indistinguishability} and the \emph{view} of an entity.
\begin{definition}[Indistinguishability]
  Let $\eta$ be a security parameter and $X_\eta$ and $Y_\eta$ two
  distributions. We say that $X_\eta$ and $Y_\eta$ are
  \emph{indistinguishable}, denoted $X_\eta \equiv
  Y_\eta$, if for every probabilistic  distinguisher
  $\D$ we have:
  $$  \Pr[x \leftarrow X_\eta : 1 \leftarrow \D(x)] - \Pr[y
    \leftarrow Y_\eta:1 \leftarrow \D(y)] = 0 
  $$ %
\end{definition}
\begin{definition}[view of a party]
  Let $\pi(I)$ be an $n$-parties protocol for the entities $(P_i)_{1\leq i \leq n}$
  using inputs $I=(I_i)_{1\leq i \leq n}$.  The view of a party $P_i(I_i)$
  (where $1\leq i \leq n$) during an execution of $\pi$, denoted
  $\textsc{view}_{\mathsf{\pi}(I)}(P_i(I_i))$, is the set of all values sent
  and received by $P_i$ during the protocol.  
\end{definition}
To prove that a party $P$ learns nothing during execution of the
protocol, we show that $P$ can run a \emph{simulator} algorithm that
simulates the protocol, such that $P$ (or any polynomially bounded
algorithm) is not able to differentiate an execution of the simulator
and an execution of the real protocol.  The idea is the following:
since the entity $P$ is able to generate his view using the simulator
without the secret inputs of other entities, $P$ cannot extract any
information from his view during the protocol.  This notion is
formalized in Definition~\ref{security}.
\begin{definition}[security with respect to semi-honest behavior]
  \label{security}
   Let $\pi(I)$ be an $n$-parties protocol between the entites $(P_i)_{1\leq i \leq n}$
   using inputs $I=(I_i)_{1\leq i \leq n}$.  We say
   that $\pi$ is \emph{secure in the presence of
     semi-honest adversaries} if for each $P_i$ (where $1\leq i \leq
   n$) there exists a protocol $\mathsf{Sim}_i(I_i)$ where $P_i$  interacts with a polynomial time algorithm $S_{i}(I_i)$
   such that:
  $$ \textsc{view}_{\mathsf{Sim}_i(I_i)} (P_i(I_i)) \equiv
  \textsc{view}_{\mathsf{\pi}(I)} (P_i(I_i)) 
  $$
\end{definition}
\begin{theorem}\label{th:sec}
Our conspiracy Santa protocol is secure with respect to semi-honest behavior.
\end{theorem}
\begin{proof}
We denote our protocol by $\mathsf{SCS}_n(I)$ (for \emph{Secure
  Conspiracy Santa}). For all $1\leq i \leq n$, each entity $P_i$ has
the input $I_i=(n,B,p_i)$, where $I=(I_i)_{1\leq i \leq n}$, $p_i$ is
the balance of each of the $n$ participants.  For all $1\leq i \leq
n$, we show how to build the protocol $\mathsf{Sim}_{i}$ such that:
   $$ \textsc{view}_{\mathsf{Sim}_i(I_i)} (P_i(I_i)) \equiv
  \textsc{view}_{\mathsf{SCS}_n(I)} (P_i(I_i)) 
  $$
   $\mathsf{Sim}_{1}$ is given in Simulator~\ref{alg:sim}, and
   $\mathsf{Sim}_{i}$ for  $1<i\leq n$ is given in
   Simulator~\ref{alg:simi}.

\makeatletter
\renewcommand{\ALG@name}{Simulator}
\makeatother
\begin{algorithm}[htb]
\caption{Algorithm $S_1$ of the protocol $\mathsf{Sim}_1(I_1)$.}\label{alg:sim}
\begin{algorithmic}[1]
  \Require $S_1$ knows $I_1=(n,B,p_1)$
  \State $S_1$ receives $t_1$ \cent{} from $P_1$; \label{si:tfromp}
  \If{$0\leq (p_1-t_1)$} \label{si:ttop}
  \State{$S_1$ sends $(B-(p_1-t_1))$ \cent{} to $P_1$;}
  \ElsIf{$(p_1-t_1) <0$}
  \State{$S_1$ sends $(B-((t_1-p_1)\mod B))$ \cent{} to $P_1$;}
  \EndIf
  
  \For{$j=1$ \To $n-1$}
        \State $S_1$ sends $B$ \cent{} to the shared anonymous address; \label{si:put}
  \EndFor
  
  \If{$0\leq (p_1-t_1)$}
  \State{$x=n$;}
  \ElsIf{$(p_1-t_1) <0$}
  \State{$x=n+\frac{(p_1-t_1)-((t_1-p_1)\mod B )}{B}$;}
  \EndIf
  
  \For{$j=1$ \To $x$}
        		\State $S_1$ makes the shared anonymous address
                        pay $B$ \cent{} to an anonymous
                        address; \label{si:take}
  \EndFor
\end{algorithmic}
\end{algorithm}

\begin{algorithm}[htb]
\caption{Algorithm $S_i$ of the protocol $\mathsf{Sim}_i(I_i)$, where $1<i\leq n$.}\label{alg:simi}
\begin{algorithmic}[1]
  \Require $S_i$ knows $I_1=(n,B,p_i)$ 
  \State $t_{i-1} \stackrel{\$}{\longleftarrow}\left[1..B\right]$ ; \label{randt}
  \State{$S_i$ sends $t_{i-1}$ \cent{} to $P_i$;}
  \State{$S_i$ receives $t_{i}$ \cent{} from $P_i$;}\label{si:tfrompnn}
  \For{$j=1$ \To $n-1$}
        \State $S_i$ sends $B$ \cent{} to the shared anonymous address;  \label{si:puti}
  \EndFor
  \State{$x=n+\frac{p_i+t_{i-1}-t_i -B}{B}$;}
  \For{$j=1$ \To $x$}
        		\State $S_i$ makes the shared anonymous address
                        pay $B$ \cent{} to an anonymous
                        address;  \label{si:takei}
  \EndFor
\end{algorithmic}
\end{algorithm}

   We first show that the view of $P_1$ in the real protocol $\mathsf{SCS}_n$ is the same as in the protocol $\mathsf{Sim}_1$:
   \begin{itemize}
        \item At Instruction~\ref{si:tfromp} of Simulator~\ref{alg:sim}, $S_1$ receives $t_1$ \cent{} from $P_1$ such that $1\leq t_1 \leq B$, as at Instruction~\ref{i:tfromp} of Protocol~\ref{alg:firstround}.
        \item At Instruction~\ref{i:ttop} of Protocol~\ref{alg:firstround}, $P_n$ sends  $t_n$ \cent{} to $P_1$ such that:
        \begin{itemize}
                \item[$\bullet$] $1\leq t_n \leq B$
                \item[$\bullet$] The balance of $P_1$ is a multiple of $B$.
        \end{itemize}
        We show that these two conditions hold in the simulator.
        At Instruction~\ref{si:ttop} of Protocol~\ref{alg:firstround}, the balance of $P_1$ is $(p_1-t_1)$. 
        \begin{enumerate}
            \item If the balance is positive, then $0 \leq (p_1-t_1) < B$ and $S_1$ sends $B-(p_1-t_1)$ \cent{} to $P_1$. We then have:
            \begin{itemize}
                \item[$\bullet$] $1 \leq B- (p_1-t_1) \leq B$
                \item[$\bullet$] The balance of $P_1$ is $B- (p_1-t_1) + (p_1-t_1) =B$ which is multiple of $B$.
            \end{itemize}
            \item If the balance is negative, then  $S_1$ sends $(B-((t_1-p_1)\mod B))$ \cent{} to $P_1$. We then have:
            \begin{itemize}
                \item[$\bullet$] $1\leq B-((t_1-p_1)\mod B) \leq B$
                \item[$\bullet$] The balance of $P_1$ is:
                $B-((t_1-p_1)\mod B) + (p_1-t_1) = B + \left\lfloor \frac{p_1-t_1}{B} \right\rfloor \cdot B =   \left( \left\lfloor \frac{p_1-t_1}{B}  \right\rfloor +1\right) \cdot B $,
                which is a multiple of $B$.
            \end{itemize}
        \end{enumerate}
        \item At Instruction~\ref{si:put} of Simulator~\ref{alg:sim}, $S_1$ sends $B$ \cent{} to the shared anonymous address $(n-1)$ times, and $P_1$ sends $B$ \cent{} to the shared anonymous address $1$ time, so together they send $B$ \cent{} $n$ times to the shared anonymous address, as at Instruction~\ref{i:put} of Protocol~\ref{alg:secondround}.
        \item At Instruction~\ref{i:take} of Protocol~\ref{alg:secondround}, the users  make the shared anonymous address
                        pay $B$ \cent{} to $n$ anonymous
                        addresses. 
        At Instruction~\ref{si:take} of Simulator~\ref{alg:sim}, the balance of $P_1$ is:
        \begin{itemize}
            \item $0$ if $0\leq (p_1-t_1)$ (because $P_1$ had $B$ \cent{} and sent $B$ \cent{} to the shared address).
            \item Otherwise, the balance of $P_1$ is $B-((t_1-p_1)\mod B) + (p_1-t_1)-B = ((t_1-p_1)\mod B) + (p_1-t_1)$. Hence $P_1$ receives $B$ \cent{} from the shared anonymous address $\left|\frac{((t_1-p_1)\mod B) + (p_1-t_1)}{B}\right|$ times, and $S_1$  receives $B$ \cent{} from the shared anonymous address $n+\frac{((t_1-p_1)\mod B) + (p_1-t_1)}{B}$
            times. 
            We note that $((t_1-p_1)\mod B) + (p_1-t_1) \leq 0$ because $(p_1-t_1) \leq 0$ and $((t_1-p_1)\mod B) \leq - (p_1-t_1)$.
            Finally, $P_1$ and $S_1$  make the shared anonymous address
                        pay $B$ \cent{} to $n$ anonymous
                        addresses because: 
                        \begin{small}
                          $$ n+\frac{(t_1-p_1)\mod B + (p_1-t_1)}{B} +\left|\frac{(t_1-p_1)\mod B + (p_1-t_1)}{B}\right| = n$$
                        \end{small}

        \end{itemize}
        \end{itemize}
        Finally, we deduce that the  view of $P_1$ in the real protocol $\mathsf{SCS}_n$ is the the same as in the simulator $\mathsf{Sim}_1$: 
         $$ \textsc{view}_{\mathsf{Sim}_1(I_1)} (P_1(I_1)) \equiv
  \textsc{view}_{\mathsf{SCS}_n(I)} (P_1(I_1)) 
  $$

   We then show that the view of $P_i$ in the real protocol $\mathsf{SCS}_n$ is the same as in the protocol $\mathsf{Sim}_i$ for any $1 \leq i \leq n$:
   \begin{itemize}
   
        \item At instruction~\ref{i:tfromp} and~\ref{i:tfrompn} of Protocol~\ref{alg:firstround}, each user $P_i$ receives $t_{i-1}$ \cent{} from $P_{i-1}$ for any $1 \leq i \leq n$ such that $1\leq t_{i-1}\leq B$. We note that each $t_{i-1}$ depends on the value $t_1$ chosen by $P_1$. Moreover, $t_1$ comes form a uniform distribution and acts as a one-time pad on the values $t_{i-1}$, i.e., it \emph{randomizes} $t_{i-1}$ such that $P_i$ cannot distinguish whether $t_{i-1}$ was correctly generated or comes from the uniform distribution on $\{ 1,\ldots , B\}$.  At instruction~\ref{randt} of Simulator~\ref{alg:simi}, $S_i$ chooses $t_{i-1}$ at random in the uniform distribution on $\{ 1,\ldots , B\}$ and sends $t_{i-1}$ to $P_i$.
    
        \item At Instruction~\ref{si:tfrompnn} of Simulator~\ref{alg:simi}, $S_i$ receives $t_i$ \cent{} from $P_i$ such that $1\leq t_1 \leq B$, like at Instruction~\ref{i:tfrompn} of Protocol~\ref{alg:firstround}.
         \item At Instruction~\ref{si:puti} of Simulator~\ref{alg:simi}, $S_i$ sends $B$ \cent{} to the shared anonymous address $(n-1)$ times, and $P_i$ sends $B$ \cent{} to the shared anonymous address $1$ time, so together they send $B$ \cent{} $n$ times to the shared anonymous address, as at Instruction~\ref{i:put} of Protocol~\ref{alg:secondround}.
         \item At Instruction~\ref{i:take} of Protocol~\ref{alg:secondround}, the users  make the shared anonymous address
                         pay $B$ \cent{} to $n$ anonymous
                         addresses. 
         At Instruction~\ref{si:takei} of Simulator~\ref{alg:simi}, the balance of $P_i$ 
         is $p_i+t_{i-1} - t_i -B $. Hence $P_i$ receives $B$ \cent{} from the shared anonymous address $\left| \frac{p_i+t_{i-1} - t_i -B }{B} \right|$ times, and $S_i$  receives $B$ \cent{} from the shared anonymous address $n+\frac{p_i+t_{i-1} - t_i -B }{B}$
             times. We note that $p_i+t_{i-1} - t_i -B \leq 0$; indeed, we have $t_i=(p_i+t_{i-1}) \mod B$ (Instruction~\ref{i:piti} of Protocol~\ref{alg:firstround}). Since $p_i\leq B$ and ${t_{i-1}\leq B}$, then we have $(p_i+t_{i-1}) - t_i \leq B$, so we have $p_i+t_{i-1} - t_i -B \leq 0$.  Finally, $P_i$ and $S_i$  make the shared anonymous address
                         pay $B$ \cent{} to $n$ anonymous
                         addresses because: 
                         $$ n+\frac{p_i+t_{i-1} - t_i -B }{B} +\left|\frac{p_i+t_{i-1} - t_i -B }{B}\right| = n$$
        \end{itemize}
        Finally, to conclude the proof, we deduce that for all $1\leq i \leq n$ the  view of $P_i$ in the real protocol $\mathsf{SCS}_n$ is the the same as in the simulator $\mathsf{Sim}_i$: 
         $$ \textsc{view}_{\mathsf{Sim}_i(I_i)} (P_i(I_i)) \equiv
  \textsc{view}_{\mathsf{SCS}_n(I)} (P_i(I_i)).\vspace{-15pt}$$
\hfill  \qed\end{proof}

\subsection{Physical Variant}\label{ssec:physical}

If one does not wish to use cryptocurrencies, one can use the following
physical variant of the protocol.
In the first round each participant needs to transfer some money to
another participant using a private channel. A simple physical
solution is that they meet and perform the transfer face to face, while ensuring that nobody spies on them. %
For the second round, the balance of all participants is a multiple of
$B$ \cent.  During the first part of this algorithm, everyone puts an
envelope containing $B$ \cent{} onto a stack that is in a secure room.
By \emph{secure room}, we mean a place where no other participants can spy
what is going on inside.
In the second part all participants enter this secure room one after the other
and do the following according to their balance:
\begin{itemize}
\item If the balance is $0$ then the participant does nothing. 
\item If the balance is a multiple $k$ of $B$ \cent, the participant
  takes $k$ envelopes from the top of the stack, opens them and
  collects the corresponding $k \cdot B$ \cent{}.
  Then he places, in each of the now empty $k$ envelopes,
  a piece of paper that have the same shape and weight as a the
  $B$ \cent{}. These envelopes are placed under the stack of
  envelopes.
\end{itemize}
This method allows everyone to collect his money without revealing to
the other ones how much they have taken.

We show that this protocol is secure with respect to semi-honest
behavior.
For this, we physically simulate the protocol for any participant.
We first note that the first round of the protocol is the same as Protocol~\ref{alg:firstround}, so this round can be simulated exactly as in the proof of Theorem~\ref{th:sec}.
We simulate the second round for any participant as follows. During
the first part of the algorithm, the simulator enters $n-1$ times
the secure room and puts an
envelope containing $B$ \cent{} onto the stack. When it is his turn, the participant enters the room and puts an
envelope containing $B$ \cent{} onto the stack. Finally, there are $n$ envelopes containing $B$ \cent{} on a stack.
In the second part the simulator enters the room $n-1$ times and does
nothing. When it is his turn, the participant enters the room and
takes $k$ envelopes from the top of the stack, opens them and collects
the corresponding $k \cdot B$ \cent{} as in the real protocol, where $0\leq k\leq n$. Since each of the $n$ envelopes contains $B$ \cent{}, the simulation works for any $0\leq k\leq n$. 
  
We deduce that the view of the participant during the simulation is
the same as during the real protocol, which implies that our physical
protocol is secure with respect to semi-honest behavior.

\begin{remark}
This physical protocol mimics exactly the solution using
cryptocurrencies. One advantage, though, of the physical world is that it is
easier to perform transactions with $0$ \cent{}.
Therefore there exists a simpler solution for the second
round, where creditors do not have to give $B$ \cent{} in advance:
if the participant is in debt he puts an envelope containing $B$
\cent{} onto the stack, otherwise he puts an envelope containing a
piece of paper under the stack.

The first and third rounds are not modified, and the simulator for the
security proof is not modified either.
\end{remark}

\section{A Faster Protocol}\label{sec:faster}
To the price of having much larger transactions, there exist a
protocol for conspiracy Santa that requires less transactions: $2
\cdot n+1$ instead of $3 \cdot n$.

\subsection{Merging the first two rounds}

We now propose a variant of the cryptocurrency protocol where the
initialization phase and the third round are unchanged but the first
and second round are merged.

The idea is that the participants will round their debts or credits so
that the different amounts become indistinguishable, and at the same
time they will give $B$~\cent{} to the next player: this is what they
would have given to the Piggy bank in the second round. 
At the end the first player will receive the amount needed to round
his debt or credit as previously, {\em plus} $n-1$ times $B$~\cent{}. 
It is then sufficient for him to give $n \cdot B$~\cent{} to the Piggy bank,
for the third round to take place as previously. 
The details are described
in Protocols~\ref{alg:merged} and~\ref{alg:onlythird}.

\begin{algorithm}[htb]
\caption{Secure rounding with an extra multiple}\label{alg:merged}
\begin{algorithmic}[1]
  \Require An upper bound $B$ on the value of any gift;
  \Require Each one of $n$ participants knows his integer balance $p_i$;
  \Require $\sum_{i=1}^n p_i=0$.
  \Ensure Each one of $n$ participants has a new balance $p_i$, either
  $0$, $B$ or a negative multiple of $B$;
  \Ensure  $\sum_{i=1}^n p_i=0$;
  \Ensure The credit of the piggy bank is $nB$~\cent{};
  \Ensure The protocol is zero-knowledge.
  \State{$P_1$}: $t_1\stackrel{\$}{\longleftarrow}\left[0..B-1\right]$ \label{fPoneBal}
  uniformly sampled at random;
  \State{$P_1$}: $p_1 = p_1 - 1-t_1$;
  \State $P_1$ sends $1+t_1$~\cent{} to $P_2$; \hfill \Comment{Random
    transaction $1..B$ on a secure channel} \label{fPoneSend}
  \State{$P_2$}: $p_2 = p_2 + 1 + t_1$;
  \For{$i=2$ \To $n-1$}
  	\State{$P_i$}: $t_i = p_i - 1\mod B$;\hfill\Comment{$0 \leq t_i
          \leq B-1$} 
	\State{$P_i$}: $p_i = p_i - 1 - t_i - (i-1) \cdot B$;
	\State $P_i$ sends $(1+t_i+(i-1) \cdot B)$~\cent{} to $P_{i+1}$; \newline
      ~\hfill \phantom{d} \Comment{Transaction $(i-1)\cdot B+1..i\cdot B$ on a secure channel} \label{fPiSend}
	\State{$P_{i+1}$}: $p_{i+1} = p_{i+1} + 1+t_i + (i-1) \cdot B$;
  \EndFor
  \State{$P_n$}: $t_n = p_n - 1\mod B $;\hfill\Comment{$0 \leq t_n \leq B-1$}
  \State{$P_n$}: $p_n = p_n - 1 - t_n - (n-1) \cdot B$;
  \State $P_n$ sends $(1+t_n + (n-1) \cdot B)$~\cent{} to $P_1$; \hfill
  \Comment{Transaction $(n-1)\cdot B+1..n\cdot B$ on a secure channel} \label{fPnSend}
  \State $P_1$ sends $n \cdot B$~\cent{} to the piggy
  bank;\hfill\Comment{Public transaction of $nB$} \label{fnBtoPB}
  \State{$P_1$}: $p_1 = p_1 + 1 + t_n-B$;
\end{algorithmic}
\end{algorithm}

\begin{algorithm}[htb]
\caption{Secure recovering from the piggy bank}\label{alg:onlythird}
\begin{algorithmic}[1]
  \Require An upper bound $B$ on the value of any gift;
  \Require $n$ participants each with a balance $p_i$, either
  $0$, or a negative multiple of $B$, their sum being $-nB$;
  \Require The piggy bank has a credit of $nB$~\cent{}.
  \Ensure All balances are zero;
  \Ensure The protocol is zero-knowledge.
  \ParFor{$i=1$ \To $n$} \label{frec}
	\If{$p_i<0$}\Comment{Creditors recover their assets}
  		\ParFor{$j=1$ \To $\frac{-p_i}{B}$}\label{frecp}
        		\State $P_i$ makes the shared anonymous address
                        pay $B$\cent{} to one of his own anonymous
                        addresses; 
                        		\hfill \Comment{Public transaction of $B$}
  		\EndParFor
        	\State{$P_i$}: $p_i = 0$.
        \EndIf
  \EndParFor
\end{algorithmic}
\end{algorithm}

\begin{theorem} For $n$ participants, Protocols~\ref{alg:setup}, \ref{alg:merged},
  \ref{alg:onlythird} are correct and, apart from the setup, require
  $2 \cdot n+1$ transactions.
\end{theorem}
\begin{proof}
Including the piggy bank, all the transactions are among participants,
therefore the sum of all the debts and credits is invariant and
zero. There remains to prove that in the end of the protocol all the
debts and credits are also zero.  Any initial debt for any gift is
strictly lower than $B/(n-1)$ and the largest debt is thus strictly
lower than $B$~\cent{}, that is all integral debts are between $0$ and
$B-1$.  Then, during the first round, all participants, except $P_1$,
still round their credits or debts to multiples of~$B$: for $i\geq 2$,
$p_i = (p_i-1)-(p_i-1 \mod B)-(i \cdot 1) \cdot B \equiv 0\mod B$.
Also, every participant sends more than he received:
$1+t_i+(i-1)\cdot B<1+t_{i+1}+i \cdot B$, so any creditor remains creditor.
Furthermore, any debtor will thus either be at zero after the first
round or at a credit of exactly $B$~\cent{}: any debtor $(i+1) \neq 1$
has an initial debt $1 \leq p_{i+1} \leq B$; then he receives
$1+t_{i}+(i-1)\cdot B$ so that his new debt is between $2+(i-1) \cdot B$ and
$(i+1)\cdot B$; finally he sends the residue mod $B$ plus $i \cdot B$, so
that his credit must be either $B$ or $0$.  All the money sent is
immediately added to the debt of the receiver and removed from that of
the sender (including an aggregation for $P_1$ at the end,
$(n-1)\cdot B-n\cdot B=-B$) so, overall, including the piggy bank, the sum of all
debts and credits remains zero.  We just showed that except $P_1$ all
the credits or debts are now multiples of $B$ so that of $P_1$ must
also be a multiple $B$.  Further, if $P_1$ is a debtor then $1 \leq
p_1 \leq B$. He sends $1+t_1$, $n\cdot B$ and receives
$1+t_n+(n-1)\cdot B$ so that his new balance is $2-2\cdot B \leq
p_1-t_1-n\cdot B+t_n+(n-1)\cdot B \leq B-1$.  As we have shown that
this balance must be a multiple of $B$, then either his balance is
zero or he has a credit of $B$.  Therefore, at the end of
Protocol~\ref{alg:merged}, only the piggy bank has a debt. Since the
piggy bank received exactly $n \cdot B$~\cent{}, exactly $n$
transactions of~$B$~\cent{} will make it zero and the invariant
ensures that, after Protocol~\ref{alg:onlythird}, all the creditors
must be zero too.  \qed\end{proof}

\begin{example}\label{ex:faster}
We now have a look at the algorithm for our example with \Alice, \Bob,
\Carole and \Dan. As in Example~\ref{ex:algo}, the initial balance
vector is $[5,48,-73,20]$ and the upper bound is $B=50$~\euro{}.  For
the first round, \Alice randomly selects $0 \leq{} t_1=11 < 50$ and
makes a first private transaction of $1+t_1=12$~\euro{} to \Bob.  \Bob
then makes a private transaction of $1+(12+48 - 1\mod
50)+50=60$~\euro{} to \Carole; \Carole makes a private transaction of
$1+(60-73-1 \mod 50)+2 \cdot 50 = 137$~\euro{} to \Dan; who makes a private
transaction of $1+(137+20-1 \mod 50)+3 \cdot 50 = 157$~\euro{} to \Alice and
the balance vector is $[150,0,-150,0]$ Finally \Alice gives
$200$~\euro{} to the piggy bank so that the player's balance vector
becomes $[-50,0,-150,0]$.  Indeed, for instance, \Bob had a balance of
$48$~\euro{}, received $12$~\euro{} from \Alice and sends $60$~\euro{}
to \Carole, hence his new balance is $48+12-60=0$~\euro{}.  All these
transactions are represented in Figure~\ref{fig:1roundmerged}.
Finally there are four $50$~\euro{} transactions, one to an address
controlled by \Alice and three to (different) addresses controlled by
\Carole. This last round is exactly the one illustrated in
Figure~\ref{fig-2round}, right.  Note that we have $n=4$ participants
and exactly $1+(n-2)+1+1=5$ transactions in Protocol~\ref{alg:merged}
and $4$ transactions in Protocol~\ref{alg:onlythird}, hence the total
is $9=5+4=2\cdot 4 + 1 = 2 \cdot  n + 1$.
\end{example}

\begin{figure}[htb]
\begin{center}
\resizebox{\textwidth}{!}{  \begin{tikzpicture}
    \node (A) at (0, 5) {\A: 5};
    \node (B) at (5, 5) {\B: 48};
    \node (C) at (5, 0) {\C: -73};
    \node (D) at (0, 0) {\X: 20};
    \node (BA) at (2.5,2.5) {Piggy Bank};

    \draw[thick,-latex] (A) .. controls (2,6) and (4,6) .. (B);
    \draw[thick,-latex] (C) .. controls (2,-1) and (1,-1) .. (D);
    \draw[thick,-latex] (D) .. controls (-1,2) and (-1,4) .. (A);
    \draw[thick,-latex] (B) .. controls (6,4) and (6,2) .. (C);
    \draw[->,thick,-latex] (A) -> (BA) node[midway,left=2mm] {200~\euro};
    
    \node (R) at (-1,5.5) {$t_1 =11$};
    \node (AB) at (2.5,6.25) {$1+11=12$~\euro{}};
    \node (BB) at (7.5,5.5) {$12+48-1=9 \mod 50$};
    \node (BC) at (7.5,2.5) {$1+9+50=60$~\euro{}};
    \node (CC) at (7.5,-0.5) {$60-73-1=36 \mod 50$};
    \node (CD) at (2.5,-1.5) {$1+36+100=137$~\euro{}};
    \node (DD) at (-2.5,-0.5) {$137+20-1=6 \mod 50$};
    \node (DA) at (-2.5,2.5) {$1+6+150=157$~\euro{}};
    
  \end{tikzpicture}
}\end{center}
\caption{Merged round of Example~\ref{ex:faster} (starting at $t_1=11$,
  ending with $5-12+157=0\mod 50$, and a final transaction of $200$~\euro{} to the Piggy bank).}\label{fig:1roundmerged}
\end{figure}
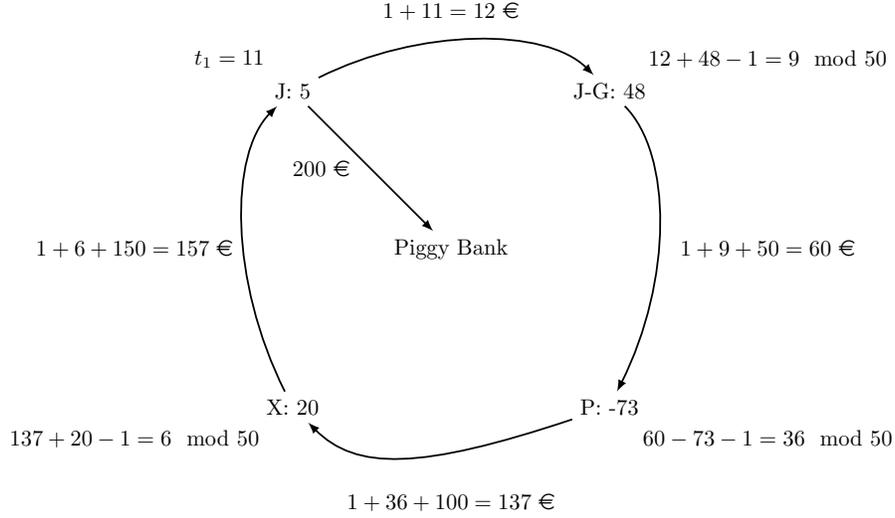

\subsection{Security Proof of the Faster Variant} 

\begin{theorem}\label{th:fsec}
Our faster conspiracy Santa protocol is secure with respect to semi-honest behavior.
\end{theorem}
\begin{proof}
We denote our protocol by $\mathsf{FSCF}_n(I)$ (for \emph{Fast Secure Conspiracy Santa}). For all $1\leq i \leq
   n$, each entity $P_i$ has the input $I_i=(n,B,p_i)$, where 
   $I=(I_i)_{1\leq i \leq n}$.
   For all $1\leq i \leq
   n$, we show how to build the protocol  $\mathsf{Sim}'_{i}$
   such that:
   $$ \textsc{view}_{\mathsf{Sim}'_i(I_i)} (P_i(I_i)) \equiv
  \textsc{view}_{\mathsf{FSCF}_n(I)} (P_i(I_i)) 
  $$
   $\mathsf{Sim}'_{1}$ is given in Simulator~\ref{alg:fsim}, and
   $\mathsf{Sim}'_{i}$ for  $1<i\leq n$ is given in
   Simulator~\ref{alg:fsimi}.

\makeatletter
\renewcommand{\ALG@name}{Simulator}
\makeatother
\begin{algorithm}[htb]
\caption{Algorithm $S_1$ of the protocol $\mathsf{Sim}'_1(I_1)$.}\label{alg:fsim}
\begin{algorithmic}[1]
  \Require $S_1$ knows $I_1=(n,B,p_1)$
  \State $S_1$ receives $(1+t_1)$ \cent{} from $P_1$; \label{si:ftfromp}
  \If{$1\leq (p_1-t_1)$} \label{si:fttop}
  \State{$S_1$ sends $(1 + B - (p_1-t_1) + (n-1)\cdot B))$ \cent{} to $P_1$;}) \label{fsPoneRecPos}
  \ElsIf{$(p_1-t_1) \leq 0$}
  \State{$S_1$ sends $(1 + ((t_1-p_1)\mod B) + (n-1)\cdot B))$ \cent{} to $P_1$;}\label{fsPoneRecNeg}
  \EndIf

  \If{$1\leq (p_1-t_1)$}
  \State{$x=n$;}
  \ElsIf{$(p_1-t_1) \leq 0$}
  \State{$x=n+\frac{(p_1-t_1)+((t_1-p_1)\mod B )-B}{B}$;}
  \EndIf
  
  \For{$j=1$ \To $x$}
        		\State $S_1$ makes the shared anonymous address
                        pay $B$ \cent{} to an anonymous
                        address; \label{si:ftake}
  \EndFor
\end{algorithmic}
\end{algorithm}

\begin{algorithm}[htb]
\caption{Algorithm $S_i$ of the protocol $\mathsf{Sim}'_i(I_i)$, where $1<i\leq n$.}\label{alg:fsimi}
\begin{algorithmic}[1]
  \Require $S_i$ knows $I_1=(n,B,p_i)$ 
  \State $t_{i-1} \stackrel{\$}{\longleftarrow}\left[0..B-1\right]$ ; \label{frandt}
  \State{$S_i$ sends $(1+t_{i-1}+(i-2) \cdot B)$ \cent{} to $P_i$;}
  \State{$S_i$ receives $(1+t_{i} + (i-1) \cdot B)$ \cent{} from $P_i$;}\label{si:ftfrompnn}
  
        \State $S_i$ sends $n\cdot B$ \cent{} to the shared anonymous address;  \label{si:fputi}
  
  \State{$x=n+\frac{p_i+t_{i-1}-t_i -B}{B}$;}
  \For{$j=1$ \To $x$}
        		\State $S_i$ makes the shared anonymous address
                        pay $B$ \cent{} to an anonymous
                        address;  \label{si:ftakei}
  \EndFor
\end{algorithmic}
\end{algorithm}

   We first show that the view of $P_1$ in the real protocol $\mathsf{FSCS}_n$ is the same as in the protocol $\mathsf{Sim}'_1$:
   \begin{itemize}
        \item At Instruction~\ref{si:ftfromp} of Simulator~\ref{alg:fsim}, $S_1$ receives $(1+t_1)$ \cent{} from $P_1$ such that $0\leq t_1 \leq B-1$, as at Instruction~\ref{fPoneSend} of Protocol~\ref{alg:merged}.
        \item At Instruction~\ref{fPnSend} of Protocol~\ref{alg:merged}, $P_n$ sends  $(1+t_n + (n-1)\cdot B)$ \cent{} to $P_1$ such that:
        \begin{itemize}
                \item $0\leq t_n \leq B $
                \item The balance of $P_1$ is a multiple of $B$.
        \end{itemize}
        We show that these two conditions hold in the simulator.
        At Instruction~\ref{fPoneBal} of Protocol~\ref{alg:merged}, the balance of $P_1$ is $(p_1- 1 -t_1)$. 
        \begin{enumerate}
            \item If $(p_1-t_1) \geq 1$, then $1 \leq (p_1-t_1) \leq  B$ and $S_1$ sends $ 1 + (B -(p_1-t_1)) + (n-1)\cdot B$ \cent{} to $P_1$. We then have:
            \begin{itemize}
                \item $0 \leq B- (p_1-t_1) \leq B-1$
                \item At Instruction~\ref{fsPoneRecPos} of Simulator~\ref{alg:fsim}, The balance of $P_1$ is $p_1 -1-t_1   +1+ (B - (p_1-t_1)+ (n-1)\cdot B =n\cdot B$, as at Instruction~\ref{fPnSend} of Protocol~\ref{alg:merged}, which is multiple of $B$.
            \end{itemize}
            \item If $(p_1-t_1)\leq 0$, then   $S_1$ sends $1+ ((t_1-p_1)\mod B) + (n-1)\cdot B$ \cent{} to $P_1$. We then have:
            \begin{itemize}
                \item $0\leq ((t_1-p_1)\mod B) \leq B-1$
                \item At Instruction~\ref{fsPoneRecNeg} of Simulator~\ref{alg:fsim}, The balance of $P_1$ is:
                $p_1- 1 - t_1 + 1 + ((t_1-p_1)\mod B) + (n-1)\cdot B  = 
                (p_1- t_1 ) + ((t_1-p_1)\mod B) +(n-1)\cdot B=    \left\lceil \frac{p_1-t_1}{B}  \right\rceil \cdot B + (n-1)\cdot B $, as at Instruction~\ref{fPnSend} of Protocol~\ref{alg:merged},
                which is a multiple of $B$.
            \end{itemize}
        \end{enumerate}
        \item $P_1$ sends $n\cdot B$ \cent{} to the shared anonymous address, so his balance is a multiple of B and is at most $0$, as at the end of Protocol~\ref{alg:merged}.
        \item At Instruction~\ref{frec} of Protocol~\ref{alg:onlythird}, the users  make the shared anonymous address
                        pay $B$ \cent{} to $n$ anonymous
                        addresses. 
        At Instruction~\ref{si:ftake} of Simulator~\ref{alg:fsim}, the balance of $P_1$ is:
        \begin{itemize}
            \item $0$ if $1\leq (p_1-t_1)$ (because $P_1$ had $n\cdot B$ \cent{} and sent $n\cdot B$ \cent{} to the shared address).
            \item Otherwise, the balance of $P_1$ is $\left\lceil \frac{p_1-t_1}{B}  \right\rceil \cdot B -B$. Hence $P_1$ receives $B$ \cent{} from the shared anonymous address $x=n+\frac{(p_1-t_1)+((t_1-p_1)\mod B )-B}{B}$ $=n+\left\lceil \frac{p_1-t_1}{B}  \right\rceil  -1$ times, and $S_1$  receives $B$ \cent{} from the shared anonymous address $n-x$
            times, as in Instruction~\ref{frecp} of Protocol~\ref{alg:onlythird}. 
            Finally, $P_1$ and $S_1$  make the shared anonymous address
                        pay $B$ \cent{} to $n$ anonymous
                        addresses.
        \end{itemize}
        \end{itemize}
        Finally, we deduce that the  view of $P_1$ in the real protocol $\mathsf{FSCS}_n$ is the the same as in the simulator $\mathsf{Sim}'_1$: 
         $$ \textsc{view}_{\mathsf{Sim}'_1(I_1)} (P_1(I_1)) \equiv
  \textsc{view}_{\mathsf{FSCS}_n(I)} (P_1(I_1)) 
  $$

   We then show that the view of $P_i$ in the real protocol $\mathsf{FSCS}_n$ is the same as in the protocol $\mathsf{Sim}'_i$ for any $1 \leq i \leq n$:
   \begin{itemize}
   
        \item At instruction~\ref{fPoneSend},~\ref{fPiSend} and~\ref{fPnSend} of Protocol~\ref{alg:merged}, each user $P_i$ receives $(1+t_{i-1} + (i-2)\cdot B )$ \cent{} from $P_{i-1}$ for any $1 \leq i \leq n$ such that $0\leq t_{i-1}\leq B-1$. As in $\mathsf{Sim}_i$ (see proof of Theorem~\ref{th:sec}) $P_i$ cannot distinguish whether $t_{i-1}$ was correctly generated or comes from the uniform distribution on $\{ 0,\ldots , B-1\}$. 
    
        \item At Instruction~\ref{si:ftfrompnn} of Simulator~\ref{alg:fsimi}, $S_i$ receives $(1+t_i + (i-1) \cdot B)$ \cent{} from $P_i$ such that $0\leq t_1 \leq B-1$, like at Instruction~\ref{fPoneSend} of Protocol~\ref{alg:merged}.
         \item At Instruction~\ref{si:fputi} of Simulator~\ref{alg:fsimi}, $S_i$ sends $n\cdot B$ \cent{} to the shared anonymous address, as $P_1$ at Instruction~\ref{fnBtoPB} of Protocol~\ref{alg:merged}.
         \item At Instruction~\ref{frec} of Protocol~\ref{alg:onlythird}, the users  make the shared anonymous address
                         pay $B$ \cent{} to $n$ anonymous
                         addresses. 
         At Instruction~\ref{si:ftakei} of Simulator~\ref{alg:fsimi}, the balance of $P_i$ 
         is $p_i+1+t_{i-1} +(i-2) \cdot B - 1-t_i - (i-1)\cdot B  = p_i+t_{i-1} - t_i -B $. Hence, setting $x=n+\frac{p_i+t_{i-1} - t_i -B}{B}$, $P_i$ receives $B$ \cent{} from the shared anonymous address $n-x=-\frac{p_i+t_{i-1} - t_i -B}{B}$ times, and $S_i$  receives $B$ \cent{} from the shared anonymous address $x$
             times.   Finally, $P_i$ and $S_i$  make the shared anonymous address
                         pay $B$ \cent{} to $n$ anonymous
                         addresses, as in $\mathsf{Sim}'_i$.
        \end{itemize}
        Finally, to conclude the proof, we deduce that for all $1\leq i \leq n$ the  view of $P_i$ in the real protocol $\mathsf{FSCS}_n$ is the the same as in the simulator $\mathsf{Sim}'_i$: 
         $$ \textsc{view}_{\mathsf{Sim}'_i(I_i)} (P_i(I_i)) \equiv
  \textsc{view}_{\mathsf{FSCS}_n(I)} (P_i(I_i)).\vspace{-15pt}$$
\qed\end{proof}

\subsection{Physical Variant}

Similarly, one can adapt the physical variant as follows.
For Protocol~\ref{alg:merged}, they perform a face to face transfer. 
At the end, the first player puts $n$ envelopes containing $B$~\cent{}
onto a stack in the secure room.
The last part is unchanged, all participants enter this secure room
one after the other 
and either does nothing or takes as many envelopes as needed to zero
out his credit.

The security proof of this physical variant also uses the simulator of
Theorem~\ref{th:fsec} for Protocol~\ref{alg:merged} and that of
Section~\ref{ssec:physical} for Protocol~\ref{alg:onlythird}.

\begin{remark}
Note that in the physical word, in some sense only $2\cdot n$ transactions
are required: $n$ face to face transfers and $n$ trips to the secure
room.
\end{remark}

\section{Conclusion}\label{sec:conclusion}

In this paper we showed that the \sharingExpensesProblem{} (\SEP{}) is \NP -complete.
Moreover, we devised two privacy-preserving protocols to share expenses in a Conspiracy Santa setting where members of a group offer each other gifts.

Our protocols ensure that no participant learns the price of his gift, while reducing the number of transactions compared to a naive solution, and not relying on a trusted third party.
We formally prove the security of our protocol and propose two
variants, one relying on cryptocurrencies for anonymous payments, the
other one using physical means, such as envelopes, to achieve anonymous payments.

Our protocol can also be used to share expenses among different groups
with non-empty intersections, while still ensuring that each
participant only learns the expenses of his group(s).

The next step is to design a practical implementation using Paypal or
a cryptocurrency. We would like to both have a platform on a website
and an application on smartphone, which however requires to solve
synchronization problems between them when not using a central
server (which was one of the design goals!).
There, we also would need also to fix random time points during an allowed
interval for transactions, set up some ciphered communications (at
setup), etc.

Finally, another avenue of research to explore could be the
determination of a lower bound on the number of transactions in a
secure peer-to-peer setting.

\section*{Acknowledgements}
\acknowledgements

\section*{References}
\bibliography{biblio}

\end{document}